\documentclass[12pt]{article}

\usepackage[margin=1.2in]{geometry}
\usepackage{amssymb}
\usepackage{amsmath}
\usepackage{amsfonts}
\usepackage{amsthm}
\usepackage{amscd}
\usepackage{bm}
\usepackage{sectsty}
\usepackage{faktor}

\usepackage[colorlinks=true]{hyperref}
\hypersetup{allcolors=[rgb]{0.1,0.1,0.4}}

\usepackage{tikz}
\usepackage{tikz-cd}

\DeclareMathOperator{\coker}{coker}
\DeclareMathOperator{\im}{im}

\newcommand{\B}{\mathcal{B}}
\newcommand{\M}{\mathcal{M}}
\renewcommand{\S}{\mathcal{S}}
\newcommand{\bR}{\mathbb{R}}
\newcommand{\bH}{\mathbb{H}}
\newcommand{\bM}{\mathbb{M}}
\newcommand{\bS}{\mathbb{S}}

\newcommand{\R}{\mathcal{R}}
\newcommand{\G}{\mathcal{G}}
\newcommand{\C}{\mathcal{C}}
\newcommand{\U}{\mathbb{U}}
\newcommand{\LL}{\mathcal{L}}
\newcommand{\dcup}{\amalg}
\newcommand{\s}{\bm{s}}
\renewcommand{\r}{\bm{r}}

\newcommand{\system}{\textbf{System}}

\newcommand{\<}{\langle}
\renewcommand{\>}{\rangle}

\newcommand{\sub}{\textbf{Sub}}

\newcommand{\ra}{\rightarrow}

\newtheorem{Def}{Definition}
\newtheorem{Pro}{Proposition}
\newtheorem{Que}{Question}

\begin{document}

\title{Generativity and Interactional Effects: an Overview\thanks{The paper is based on Chapter 2 of E. M. Adam's doctoral thesis, see \cite{ADAM:Dissertation}.}}
\author{Elie Adam\thanks{Both authors are with the Laboratory for Information and Decision Systems (LIDS) at the Massachusetts Institute of Technology (MIT).  Emails: \texttt{eadam@mit.edu} and \texttt{dahleh@mit.edu}} \and Munther A. Dahleh}
\date{}

\maketitle

\begin{abstract}
  We propose a means to relate properties of an interconnected system to its separate component systems in the presence of cascade-like phenomena.  Building on a theory of interconnection reminiscent of the behavioral approach to system theory, we introduce the notion of \emph{generativity}, and its byproduct, \emph{generative effects}.  Cascade effects, enclosing contagion phenomena and cascading failures, are seen as instances of generative effects.  The latter are precisely the instances where properties of interest are not preserved or behave very badly when systems interact. The goal is to overcome that obstruction.  We will show how to extract mathematical objects from the systems, that encode their generativity: their potential to generate new phenomena upon interaction.  Those objects may then be used to link the properties of the interconnected system to its separate systems. Such a link will be executed through the use of exact sequences from commutative algebra.
\end{abstract}

\section{Introduction.}

Whenever we deal with the realm of interconnection of systems and interaction-related effects---termed, \emph{interactional effects}---we may be driven by a natural impulse to desire properties of systems, explaining such effects, that are compositional.  By compositional properties, we mean properties that are preserved or behave very well when the systems are interconnected.  Interactional effects of heavy interest, such as contagion effects and cascading failures, arise however exactly because compositionality fails in various aspects.  These phenomena are given away by systemic properties that fundamentally do not behave well under interconnection. We may then fairly expect, when aiming to understand such interactional effects, any non-trivial tangible compositional property to be fundamentally too weak to yield us something of use. But if we back away from the idea of wanting our properties to be compositional, can we recover such compositionality through other means?  This paper proposes such a means to relate properties of an interconnected system to its separate component systems in the presence of cascade-like phenomena.  Properties that are preserved or behave well under interconnection are seen to be related through a special case of that means.

This paper introduces the notion of \emph{generativity}, and its byproduct, \emph{generative effects}.  Cascade effects, enclosing contagion phenomena and cascading failures, are seen as instances of generative effects.  The converse does not have to be true unless one decides to expand the intuition that floats about cascading phenomena.  The key to understanding these effects is that they are \textbf{not} intrinsic to the system.  They rather result from an extrinsic dichotomy, a separation between what is deemed observable in the system---termed, the \emph{phenome}---and what is concealed in the system. Generative effects emerge from an interplay between the phenome and the concealed mechanisms.  This dichotomy is enforced by a map---termed, the \emph{veil}---from a space of systems to a space of phenomes.  The veil partially covers the system leaving the phenome bare, thereby concealing mechanisms in a system.  Generative effects are sustained by the veil whenever the phenome of the interconnected system cannot be explained by the phenome of the separate parts.  The mechanisms of the systems concealed under the veil interact so as to produce new observables.  

Generative effects are thus precisely the instances where the property of interest, the phenome, evolves horrendously under interconnection. The goal of the paper is to overcome that obstruction.  We will show how to extract mathematical objects from the systems, that encode their generativity: their potential to generate new observables.  Those objects may then be used to link the phenome of the interconnected system to its separate systems. Such a link will be executed through the use of exact sequences from commutative algebra.  %
The horizon goal of such a development is twofold.  We firstly aim to acquire a computational means to evaluate the behavior of systems under the presence of cascading-like phenomena. We secondly aim to gain insight into how systems interact and behave among each others.  Such insight is to be be used first for the (theoretical) analysis of interactive systems, and second for the design of systems that can desirably cope with change.

The theory of generativity may be developed on two levels.  The first is a special level where the systems in play tend to be all subsystems of a fixed system serving as a universe.  The second is a general level where such a universe system is not present.  
We will only be concerned with the special level in this paper. It is much simpler to describe.  One however needs the general level to insightfully arrive at a clear formulation of the problem, and then at a solution.  This paper can then only present a sketch of the solution, rather than the general details of it. This paper will nevertheless illustrate a working example solution in the special level.

This paper begins by a theory of interconnection reminiscent of the behavioral approach to system theory, initiated by J.\ C.\ Willems.  It proceeds to define generative effects, and elaborates examples to develop the intuition of the reader.  In a linear world, generative effects will be associated to a loss of surjectivity. We thus show how to recover such a loss, and link the phenome of the interconnected system to the separate component systems.  As not all worlds are linear, we arrive at a general solution by lifting our problems to linear problems.

On the technical end, interconnection of systems, in a linear world, will be synonymous to short exact sequences.  Generative effects are sustained by a veil if, and only if, exactness is lost once the veil is applied to the systems.  The goal is then to recover this loss of exactness.  We can extract algebraic objects from the systems, that encode their generativity, and mend the nonexact sequence into a long exact one by fitting in the objects appropriately.

\section{Interconnection and interaction of systems.}

We cannot have interactional effects without a notion of interaction or interconnection.  The notion of interconnection in the special theory is rather simple, yet inclusive.  The systems are elements of a set $S$, and interaction is an operation $\vee: S \times S \ra S$ that is:
\begin{itemize}\setlength\itemsep{0in}
  \item[I.1.] Associative i.e., $(s\vee s') \vee s'' = s \vee (s' \vee s'')$
  \item[I.2.] Commutative i.e., $s\vee s' = s' \vee s$
  \item[I.3.] Idempotent i.e., $s \vee s = s$
\end{itemize}
A semilattice $(S,\vee)$ is a set $S$ equipped with an operation $\vee$ that is associative, commutative and idempotent.  The operation $\vee$ is termed the \emph{join} of the semilattice.
      
\begin{Def}  A system is an element of a semilattice $(S,\vee)$. The systems $s$ and $s'$ in $S$ are interconnected to give the system $s \vee s'$, their join, in $S$.
\end{Def}

 A system $s$ is defined to be a subsystem of $s'$ whenever $s \vee s' = s'$.  The system $s = s_1 \vee \cdots \vee s_m$ is then the system that amounts from the interaction, or combination, of the subsystems $s_1, \cdots, s_m$.  There is only a unique way to interconnect two systems, and it is via the $\vee$ operation.

A semilattice $(S,\vee)$ induces a partial order $\leq$ on $S$ obtained by setting $s \leq s'$ if, and only if, $s\vee s' = s'$. Thus $s$ is a subsystem of $s'$ if, and only if, $s \leq s'$. The join $s \vee s'$ is then the least upper bound of $s$ and $s'$.  It is the smallest system $t$ such that $s \leq t$ and $s' \leq t$.

We provide two generic concrete interpretations, and leave other interpretations to examples in future sections.
  
\subsection{The behavioral approach to system theory.}

The behavioral approach to system theory, initiated by J.\ C.\ Willems, begins with the premise that a mathematical model acts as an exclusion law (see e.g., \cite{POL1998} and \cite{WIL2007}).  The phenomenon we wish to model produces events or outcomes that live in a given set $\U$, termed the \emph{universum}. The laws of the model (viewed descriptively) state that some outcomes in $\U$ are possible, while others are not. The model then restricts the outcomes in $\U$ to only those are allowed possible by the laws of the model.  The set of possible outcomes is then called the \emph{behavior} of a model.  We will refrain from using the term model, and replace it by the term system.

\begin{Def}[cf. \cite{POL1998} Section 1.2.1]
  A Willems system is a pair $(\U,\B)$ with $\U$ a set, called the \emph{universum}---its elements are called \emph{outcomes}---and $\B$ a subset of $\U$ called the \emph{behavior}.
\end{Def}

In case we fix a universum $\U$, the set of Willems systems $(\U,\B)$ partially ordered as $(\U, \B) \leq (\U, \B')$ if, and only if, $\B \supseteq \B'$ forms a semilattice $(S,\vee)$.  Interconnection of systems is given by the set-intersection of the behaviors, and corresponds to the join of the defined semilattice.

\begin{Pro}
 If $(\U,\B)$ and $(\U,\B')$ are Willems systems, then their interconnection $(\U, \B \cap \B')$ is given by $(\U, \B) \vee (\U, \B')$. 
\end{Pro}

The properties of the semilattice will depend on what we allow as possible behaviors.  If we consider all subsets to be possible behaviors, our semilattice $(S,\vee)$ will form a Boolean lattice.  If $\U$ is a vector space, and we consider the linear subspaces of $\U$ to be the possible behaviors, then we tend to get a semilattice that is only a modular lattice.

\subsection{Syntactical systems and descriptions.}\label{Syntax}

Another approach consists of thinking of an element of the semilattice $(D,\vee)$ as a description of a system.  Descriptions are combined through the join operation of the semilattice. The description may be in the form of a text, an equation, a diagram or any syntactical piece one might wish for.  Inasmuch as the solution set of a set of algebraic equations does not depend on their order, we have $d\vee d' = d' \vee d$.  Inasmuch as redundant algebraic equations produce no effects on the solution set, we have $d \vee d = d$.  As combining descriptions also tends to be associative, we arrive at the defining axioms of the semilattice.

\paragraph{As a simplified formalization.} 
 Let $\Sigma$ be a finite set, termed \emph{the alphabet}. We define $\Sigma^*$ to be the set of finite strings, words, or sequences, made up of elements of $\Sigma$.  If $\Sigma = \{a,b\}$, then $\Sigma^* = \{\emptyset, a, b, aa, ab, ba, bb, aaa, aab, \cdots\}$.  A description, termed \emph{language}, is a subset of $\Sigma^*$.  As not all languages may provide meaningful descriptions for our systems, we may a pick a subset $\LL$ of them.  We order $\LL$ by inclusion, and get a partial order.  We will assume that every pair of languages in $\LL$ admits a least upper-bound.  We would have then obtained a semilattice $(\LL,\vee)$. If $d$ and $d'$ are descriptions of systems in $\LL$, then $d \vee d'$ denotes the smallest language in $\LL$ that contains both $d$ and $d'$.  

 If $d$ and $d'$ are descriptions for the same system, then we expect $d \vee d'$ to be a redundant description of the same system.  Let us assume that $\LL$ is finite.
If a system admits multiple descriptions of it in $\LL$, we may take the join of all those descriptions to arrive at the maximum language in $\LL$ that describes the system. 
The collection of such maximal languages forms a subsemilattice $(S,\vee)$ of $\LL$.  The semilattice $S$ will then be our semilattice of systems. In case every system admits a unique description in $\LL$, the semilattice of systems $S$ is just $\LL$.

\section{Generativity and Interactional Effects.}

A theory of interconnection by itself will not be enough to produce interactional effects. Interconnecting two systems only gives an interconnected system.  We thus view interactional effects as fundamentally not intrinsic to the system.  They will only emerge once we set our expectation for what is deemed observable in a system.  

Let $\system$ be a semilattice of systems.  We define the phenome as that which we choose to explicitly observe from an arbitrary system in that class. A phenome may be either a property, a feature, a consequence, or even a subsystem of the system.  We generally, often non-trivially, arrive at a phenome by forgetting irrelevant information from the system. We may lightly define a phenome as the image of a system under a map $\phi: \system \ra P$ of sets.  We may choose to forget nothing at all, and get the identical whole system as a phenome. The set $P$ would then be $\system$, and $\phi$ would be the identity map. We may also choose to forget everything, and get \emph{nothing} as a phenome of a system. The set $P$ would then be a singleton set $\{*\}$, and $\phi$ would be the unique map $\system \ra \{*\}$.  Thus varying what and how much we forget from a generic system of $\system$ gives us different phenomes for the same system.

We are interested in understanding how the phenome of a certain system changes when the system is modified.  More generally, we want to understand how the phenome of systems changes when systems interact.  We are particularly interested in the situations where the phenome of the interconnected system cannot be explained by the phenome of the separate systems.  Once a phenome is declared, everything we intentionally forget from the system is declared to be concealed.  Although concealed features of a system may not emerge by themselves into the phenome, they are likely to interact with phenomes or concealed mechanisms from other systems to affect the observable phenome. Such situations are characteristic of the contagion behavior observed in societal settings and of cascading failure in various infrastructural systems.  The parts of the systems that are declared concealed may interact so as to produce more than what is expected from what is observable.  We term the effects leading to such unexplained phenomena as \emph{generative effects}.  The concealed part of the system is irrelevant to the system's phenome, but it has the potential to interact with either the phenome or concealed parts of other systems.  We term that potential \emph{generativity}.

\subsection{Veils and generativity.}

We arrive at a phenome by forgetting things from a system, by concealing them under a veil. We may then well think of a phenome as a simplified system.  The set of phenomes then  forms a semilattice $(P,\vee)$.  The join $\vee$ naturally induces a partial order $\leq$ on $P$.

\begin{Def}[Veil]
 A veil on $\system$ is a pair $(P,\phi)$ where $(P,\vee)$ is a semilattice of phenomes, and $\phi : \system \ra P$ is a map such that:
\begin{itemize}
  \item[V.1.] The map $\phi$ is order-preserving, i.e., if $s \leq s'$, then $\phi s \leq \phi s'$.
  \item[V.2.] Every phenome admits a simplest system that explains it, i.e., the set $\{s: p \leq \phi s\}$ has a (unique) minimum element for every phenome $p$.
\end{itemize}
\end{Def}
The veil is intended to hide away parts of the system, and leave other parts, the phenome, of the system bare and observable.  The axiom V.1 indicates that concealing a subsystem of a system may only yield a subphenome of the phoneme of the system.  The axiom V.2 indicates that everything one observes can be completed in a simplest way to something that extends under the veil. Generative effects occur precisely when one fails to explain the happenings through the observable part of the system.  In those settings, the things concealed under the veil would have interacted and produced observable phenomes. 

\begin{Def}[Generative Effects]
 A veil $(P,\phi)$ is said to sustain generative effects if $\phi(s \vee s') \neq \phi (s) \vee \phi(s')$ for some $s$ and $s'$.
\end{Def}

Different veils may be chosen for the same semilattice of systems.  Some will sustain generative effects and some will not.  For instance, both veils $(\system,id)$ and $(\{*\},*: \system \ra \{*\})$ do not sustain generative effects at all.  All that can be observed is explained by what is already observed.  Thus the standard intuition for systems exhibiting cascading phenomena, or contagion effects, does not stem from a property of a system.  It is rather the case that the situation admits a highly suggestive phenome and highly suggestive veil that sustains such effects.  Those effects are thus properties of the situation.  Should we change the veil, we may either increase those effects, diminish them or even make them completely go away. Such interactional effects depend only on what we wish to observe.

Our aim is then to answer the following question:
\begin{Que}
 Given a veil $(P,\phi)$ that sustains generative effects, how can we non-trivially characterize or express $\Phi(s \vee s')$ through separate information on $s$ and $s'$ (and potentially a common system $s \wedge s'$)?
\end{Que}
How can we relate the behavior of the interconnected system to its separate components?

\section{Some examples.}\label{Examples}

The aim of this section is to develop the reader's intuition on veils and generative effects.  We provide five examples. %

\subsection{Generativity is not intrinsic to the system.}\label{warmupExample} 

This example deals with a very simple---if not the simplest---instance of generative effects.  Let $U$ be a finite set.  Two proper subsets of $U$ are not equal to $U$ by definition.  Their union can however be equal to $U$.  If we set up a veil that keeps only whether a given subset of $U$ is equal to $U$ or not, the veil will then sustain generative effects.  Formally, let $\{*\}$ be a one point set. The semilattice $\system$ of systems is $(2^U,\cup)$.  The veil is $(2^{\{*\}},\phi)$ where $\phi S = \emptyset$ if $S \subsetneq U$ and $\phi S = \{*\}$ if $S = U$.  Although two subsets are not separately equal to $U$, their union can be $U$.  If $S^c$ is the complement set of $S \subsetneq U$, then:
\begin{equation*}
  \phi ( S \cup S^c) \neq \phi S \cup \phi S^c
\end{equation*}

This toy instance may be complicated by replacing $2^U$ by any semilattice $L$. Let $s,s' \in L$ be non-comparable elements, and define $\phi$ to map an element $t$ to $\{*\}$ if $t \geq s\vee s'$, and to $\emptyset$ otherwise. Such a defined veil trivially sustains generative effects for any semilattice of systems with two non-comparable elements.  It is then ill-posed to talk about a system exhibiting generative effects.  It is a property of the perspective, i.e. the veil, we choose.

\subsection{Generativity in the behavioral approach.}  

Let us consider a mega-system comprised of an interacting mixture of infrastructures (e.g., power, transportation, communication), markets (e.g., prices, firms, consumers), political entities and many individuals.  We are interested in understanding the evolution of the behavior of a subsystem of this mega-system, as changes are effected into the mega-system. Of courses changes directly effected onto the subsystem modifies the behavior. It is also the case that seemingly non-related changes causes a shift in the behavior by a successive chain of events.

Let $\bM$, $\bS$ and $\bR$ be sets such that $\bM = \bS \times \bR$. Following the behavioral approach terminology we will have $\bM$, $\bS$ and $\bR$ be the outcome space, or universum, of the \textbf{m}ega-system, the \textbf{s}ubsystem, and the \textbf{r}est (or remainder) in the mega-system that is not the subsystem of interest.  The systems will then be subsets of those universa.  Specifically, the sets $\M \subseteq \bM$, $\S \subseteq \bS$ and $\R \subseteq \bR$ denote the behavior of the mega-system, the subsystem and the rest, respectively. Although $\S$ is a subsystem of $\M$, the set $\S$ is not a subset of $\M$, but is rather a projection (or a quotient) of $\M$ onto the $\bS$-coordinate.  %
 A change in our mega-system, following the behavioral approach, is then depicted as an intersection with a change $\C \subset \bM$.  If we denote by $\pi: \bM \rightarrow \bS$ the projection onto the $\bS$-coordinate, then $\pi \M = \S$ and we generally observe:
\begin{equation*}
  \pi (\M \cap \C ) \neq \pi(\M) \cap \pi(\C).
\end{equation*}
The change $\C$ affects the subsystem $\S$ through interactions within $\R$.  If all the interactions were confined to be within $\R$, then we would have had equality for sure. In such a case, changes outside of $\S$ do not affect $\S$.

\subsubsection*{Mathematically.} The systems lattice is $(2^\bM, \cap)$, and the lattice of phenomes is then $(2^\bS, \cap)$.  The veil $\pi$ is then the projection of $\M \subseteq \bM$ onto the $\bS$-coordinate, i.e.,
\begin{equation*}
  \pi \M = \{ s \in \bS : (s,r) \in \M\}.
\end{equation*}
The map $\pi$ preserves the partial order, and every phenome $\S \subseteq \bS$ admits $\S\times \bR$ as a simplest system explaining it in $2^\bM$. The veil $\pi$ also sustains generative effects.
As a simple instance, consider $\bS = \{s,\s\}$ and $\bR = \{r, \r\}$. Let our mega-system be $\M = \{(s,r), (\s,\r)\}$, where all feasible outcomes have matched type-faces. Our system $\S$ is then $\pi \M = \{s,\s \}$.  We will now effect the following change $\C = \{(s,\r),(\s,\r)\}$, where only bold-faced $\r$ is allowed.  We then observe an inequality.  The set $\pi (\M \cap \C ) = \{\s\}$ is different than the set $\pi(\M) \cap \pi(\C) = \{s,\s\}$.  The change \emph{propagated} through $\R$ into the behavior of $\S$.

\subsubsection*{Generally.}
The universa $\bM$, $\bS$ and $\bR$ may be equipped with various mathematical structures, e.g., linear structures making them vector space.  The behaviors become subspaces of their corresponding universum.  This example may then be enriched as needed. 

\subsection{Deduction and consequences.} 

As a simplified case of this example, we will have each system consist of a three node graph. Each node in the graph can be colored either black or white, and is assigned an integer $k$ as a threshold. All nodes are white initially. A node then becomes black, if at least $k$ of its neighbors are black. Once a node is black it remains black forever. In this setting, the order of update does not affect the final set of black.  For instance, let $A$ and $B$ denote the systems on the left and right, respectively.

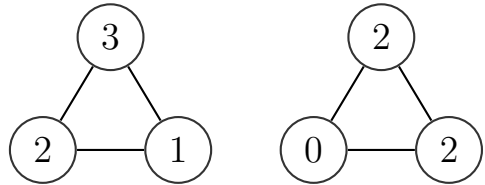
\begin{center}
\begin{tikzpicture}
  [scale=.3,auto=center,thick,every node/.style={circle,draw=black!80!white,scale=1.2}]
  \node (n1) at (0,0) {2};
  \node (n2) at (3,5)  {3};
  \node (n3) at (6,0)  {1};

  \node (n11) at (12,0) {0};
  \node (n22) at (15,5)  {2};
  \node (n33) at (18,0)  {2};

    \foreach \from/\to in {n1/n2,n2/n3,n1/n3,n11/n22,n22/n33,n11/n33}
    \draw (\from) -- (\to);
\end{tikzpicture}
\end{center}

Given our rule above, a threshold of $0$ indicates that a node \emph{automatically} becomes black. If no threshold of $0$ exists, then necessarily all nodes will remain white. Two systems interact by combining their evolution rules. The system $A\vee B$ corresponds to the graph that keeps on each node the minimum threshold between that of $A$ and $B$:
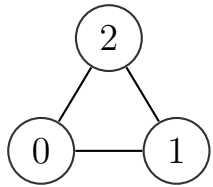
\begin{center}
\begin{tikzpicture}
  [scale=.3,auto=center,thick,every node/.style={circle,draw=black!80!white,scale=1.2}]
  \node (n1) at (0,0) {0};
  \node (n2) at (3,5)  {2};
  \node (n3) at (6,0)  {1};
  
    \foreach \from/\to in {n1/n2,n2/n3,n1/n3}
    \draw (\from) -- (\to);
\end{tikzpicture}
\end{center}

We can forget the evolution rules that are prone to interact with others by keeping from the systems only the set of final black nodes.  Indeed, every set of black nodes $S$ corresponds to a simplified system having a threshold of $0$ on the nodes in $S$ and a threshold of $\infty$ on the nodes not in $S$.  Let us denote by $\phi(A)$ and $\phi(B)$ the set of black nodes of A and B respectively.  Then the set $\phi(A)$ is empty, and the set $\phi(B)$ contains the left node.  The combination of the phenomes $\phi(A)$ and $\phi(B)$ corresponds to the union $\phi(A) \cup \phi(B)$. Such a combination may be equivalently thought of as the final set of black nodes corresponding to $\phi(A)$ and $\phi(B)$ when viewed as simplified systems.  We then arrive at the inequality:
\begin{equation*}
 \phi(A \vee B) \neq \phi(A) \cup \phi(B)
\end{equation*}
When $A$ and $B$ are combined, the left black node in $B$ interacts with the rules of $A$ to color the right node black. Both the left and the right nodes then interact with the rules of $B$ to color the middle node black.  This effect is encoded in the inequality.

\subsubsection*{Generally.}

The above case may be trivially generalized to arbitrary graphs and thresholds.  The essence of it however lies in the following idea. Informally, let $E_1, \cdots, E_n$ be statements. A statement may be thought of as expression that may either be proven or unproven to be true.  It is helpful to think of statements as theorems.  The statements however may be related.  If some are proved to be true, they may constitute a proof for other statements to be true.  The systems will then consist of a set of premises and implications.  The premises may be thought of as axioms, and the implications may be thought of as inference rules.  Some statements are initially assumed true, and the implications allow us to prove more statements to be true than initially held.  Two systems are combined by combining their premises and the implication rules.  It is intuitively clear that the interaction of implications from the systems would allows a far more powerful deduction than what is possible without interaction.

\subsubsection*{Mathematically.}

The systems and their properties are treated in \cite{ADAM:AlgebraCascadeEffects}.  Let $S$ be a set. The systems can be identified with maps $f : 2^S \ra 2^S$ satisfying:
\begin{itemize}\setlength\itemsep{0in}
\item[A.1] For all $A\subseteq S$, we have $A \subseteq fA$.
\item[A.2] If $A \subseteq B$, then $f A \subseteq f B$.
\item[A.3] For all $A\subseteq S$, we have $f f A = f A$.
\end{itemize}

We can order the maps by $f \leq g$ if, and only if, $fA \leq gA$ for all $A \subseteq S$. We then obtain a semilattice $\LL$.  If we define $\phi : \LL \ra 2^S$ to map $f$ to its least fixed-point $f(\emptyset)$, then $(2^S,\phi)$ defines a veil. This veil can be shown to sustain generative effects.

The lattice $2^S$ may also be replace by any other lattice. Of course, one may also consider subsemilattices of $\LL$ for additional variations.

\subsection{Reachability problems.}

This example may viewed through the lens of reachability problems.  A system loosely consists of a collection of states along with internal evolution dynamics.  The dynamics of system dictate whether the system may evolve from state $a$ to state $b$.  Two systems may be combined by allowing their dynamics to interact.  The interaction of dynamics would then allow the system to reach more states than what is separately reachable.

As a simplified case of this example, we will have each system consist of a digraph over four nodes. For instance, let $S$ and $S'$ denote the systems on the left and right, respectively.

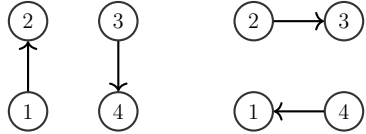
\begin{center}
\begin{tikzpicture}
  [scale=0.3,auto=center,thick,->,every node/.style={circle,draw=black!80!white,scale=0.7}]

  \node (n1) at (0,0) {1};
  \node (n2) at (0,4)  {2};
  \node (n4) at (4,0)  {4};  
  \node (n3) at (4,4) {3};

  \node (n11) at (10,0) {1};
  \node (n22) at (10,4)  {2};
  \node (n44) at (14,0)  {4};  
  \node (n33) at (14,4) {3};
  
  \foreach \from/\to in {n1/n2,n3/n4,n22/n33,n44/n11}
    \draw (\from) -> (\to);
\end{tikzpicture}
\end{center}

\noindent Two systems $S$ and $S'$ interact by combining their edges, to yield $S \cup S'$.
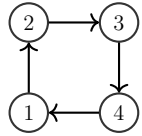
\begin{center}
\begin{tikzpicture}
  [scale=0.3,auto=center,thick,->,every node/.style={circle,draw=black!80!white,scale=0.7}]

  \node (n1) at (0,0) {1};
  \node (n2) at (0,4)  {2};
  \node (n4) at (4,0)  {4};  
  \node (n3) at (4,4) {3};
  
  \foreach \from/\to in {n1/n2,n3/n4,n2/n3,n4/n1}
    \draw (\from) -> (\to);
\end{tikzpicture}
\end{center}

The phenome of the system corresponds to the set of pairs $a\rightarrow b$ where $b$ can be reached from $a$ through a directed path. We can ignore the cases $a \rightarrow a$ as they belong to the phenome of every possible system.  Then phenome $\phi(S)$ of $S$ corresponds to $\{1 \rightarrow 2, 3 \rightarrow 4\}$, while the phenome $\phi(S')$ of $S'$ corresponds to $\{2 \rightarrow 3, 4 \rightarrow 1\}$.  The phenomes $\phi(S)$ and $\phi(S')$ are combined via set union to yield $\phi(S) \cup \phi(S')$. One can then see that:
\begin{equation*}
 \phi(S \cup S') \neq \phi(S) \cup \phi(S')
\end{equation*}
Generative effects are sustained. Indeed $\phi(S) \cup \phi(S')$ contains only four elements, whereas $\phi(S \cup S')$ contains all possible pairs. The edges in the combined graph are aligned to create paths that did not exist separately.

\subsubsection*{Generally + Mathematically.}

Let $S$ be a set.  Denote by $\textbf{Rel}(S)$ and $\textbf{Tran}(S)$ the set of relations on $S$ and transitive relations on $S$, respectively.  The set $\textbf{Rel}(S)$ froms a semilattice by defining the join to be union of sets, and $\textbf{Tran}(S)$ forms a semilattice by defining the join ot be the union of sets followed by the transitive closure.  The systems lattice is then $\textbf{Tran}(S)$, the phenomes lattice is $\textbf{Rel}(S)$ and our veil will be defined to forget the transitive property. The map $\phi$ of the veil will be the order-preserving inclusion from $\textbf{Tran}(S)$ to $\textbf{Rel}(S)$.  The defined veil sustains generative effects.

\subsubsection*{Important remark.}  The semilattice of phenomes does not have to be \emph{smaller} than the semilattice of systems.  The veil can be devised to forget properties, and thus the phenome consists of systems that do not necessarily have the forgotten property.  The space of phenomes then trivially contains the systems that do have the forgotten property.  Throwing away information from the system, leaves us with a system with less information.  But if that information was restrictive, then the space of phenomes will be \emph{greater} than that of the systems.

\subsection{Words, languages and grammars.}  

Let $\Sigma = \{a,b\}$ be an alphabet set.  A word over $\Sigma$ is a string consisting of a finite sequence of letters in $\Sigma$, e.g., $abba$, $a$, $abaab$, etc.
A system will consist of a collection of transformation rules $u \leftrightarrow v$ where $u$ and $v$ are words.  Starting from a given word $w$, a system that possesses rule $u \leftrightarrow v$ may substitute any appearance of $u$ as a subword in $w$ by a subword $v$, and vice versa.  Two systems are combined by taking the union of the rules. Fixing an initial word $w$, the phenome we are interested in is the set of words that a system may transform it to. Indeed, the scope may be much greater when systems are combined than what can be separately achieved. Generative effects will be sustained.

As a concrete instance, let the system $S$ be the rules $aa \leftrightarrow a$ and $bb \leftrightarrow b$, and the system $S'$ be the rules $ab \leftrightarrow ba$. Let us pick (and fix) $w$ to be $ab$. The phenome $\phi(S)$ is then the set of words having all $a$s on the left and all $b$s on the right. The phenome $\phi(S')$ is only the set $\{ab, ba\}$.

As two systems interact by putting their relations in common, the system $S \vee S'$ is the rules $aa \leftrightarrow a$, $bb \leftrightarrow b$ and $ab \leftrightarrow ba$. The combination of phenomes is the set union $\phi(S) \cup \phi(S')$. Generative effects are sustained as:
\begin{equation*}
  \phi(S \vee S') \neq \phi(S) \cup \phi(S').
\end{equation*}
Indeed, the phenome $S \vee S'$ contains all strings containing at least one $a$ and one $b$.

\subsubsection*{Mathematically.}

This example falls within the world of languages and grammars. We will unfortunately neglect the algebraic structure of the problem, for the purposes of this example. Let $\Sigma^*$ be the set of all words, and let $w$ be a fixed word. Every system corresponds to an equivalence relation on $\Sigma^*$. The phenome corresponds to the equivalence class containing the word $w$.  Combination of systems corresponds to closing the union relation under transitivity, and combination of phenomes is set union. Both the systems and the phenomes form semilattices.  The veil that reads the equivalence class of $w$ can be shown to sustain generative effects.

\section{The problem and the goal.}

Our aim is to understand the evolution of the phenome as systems interact.  The inequality in generative effects hinders such an understanding.  Since we are precisely interested in such effects, we are bound to live with that inequality. The question then becomes as to how we go around it.  We are interested in the phenome of the interconnected system.  We may obviously, if tractable, combine the systems and read the phenome. Such an approach, however, will yield no insight at all into the problem.

The inequality of generative effects tells us that some features of the combined phenome cannot be explained by the separate ones.  Thus we still need to extract additional information from the system.  We will then extract a mathematical object that encodes the generativity of a system: the potential of a system to produce changes in the phenome.  We can then use these objects to relate the phenome of the combined system to that of the separate subsystems. Thus, those object will summarize the required information needed to go around the inequality.  But most importantly, in most cases, the system cannot be reconstructed from the objects and the phenomes.  We are then distilling what it is that makes systems produce those effects.  In the general theory, these objects may be seen as universal in a certain sense.  The development in this paper will however be oblivious to any property those objects ought to possess.

The key to the solution is that the inequality in generative effects means that some features of the combined system's phenome cannot be \emph{explained} by the separate subsystems' phenome.

\subsection{Destroying surjectivity.}

A veil $(P,\phi)$ is said to sustain generative effects if $\phi(s \vee s') \neq \phi (s) \vee \phi(s')$ for some $s$ and $s'$.  The phenome of the separate systems is, thus, unable to explain the phenome of the interconnected system.  This inability will be formally understood as a loss of surjectivity of a certain map.  The loss of surjectivity will be key to the solution.

Let us suppose that both the systems and the phenomes are sets. More precisely, we will have $\system = (2^U,\cup)$ and $P = (2^V,\cup)$.  If $S$ and $S'$ are subsets of $U$, then there are canonical injective maps $i: S \ra S \cup S'$ and $i' : S' \ra S \cup S'$.  Let $S \dcup S'$ denote the disjoint union of $S$ and $S'$, and define $i \dcup i' :  S \dcup S' \ra S \cup S'$ to be the map $i$ on $S$ and $i'$ and $S'$.

\begin{Pro}\label{Pro:surj}
 For every every $S$ and $S'$, the map $i \dcup i'$ is a surjective map.
\end{Pro}

\begin{proof}
 An element of $S\cup S'$ belongs to either $S$, $S'$ or both.
\end{proof}

On another end, as $\phi S \subset \phi(S \cup S')$ and $\phi S' \subset \phi(S \cup S')$, we get canonical injective maps $\phi i :  \phi S \ra \phi(S \cup S')$, and $\phi i' : \phi S' \ra \phi(S \cup S')$.  The map $\phi i \dcup \phi i'$ need not always be surjective.

\begin{Pro}\label{Pro:surjG}
 The veil $\phi$ sustain generative effects if, and only if, the map $\phi i \dcup \phi i' : \phi S \dcup \phi S' \ra \phi(S \cup S')$ is not surjective for some $S$ and $S'$.
\end{Pro}

\begin{proof}
 Generative effects are not sustained if, and only if, $\phi(S \cup S') = \phi(S) \cup \phi(S')$ for all $S$ and $S'$.  If $\phi(S \cup S') = \phi(S) \cup \phi(S')$, then $\phi i \dcup \phi i' : \phi S \dcup \phi S' \ra \phi(S \cup S')$ is surjective by Proposition \ref{Pro:surj}.  Conversely, if $\phi i \dcup \phi i'$ is not surjective, then some element in $\phi(S \cup S')$ does not admit a preimage in $\phi S \dcup \phi S'$.  Therefore, $\phi(S \cup S')$ strictly contains $\phi(S) \cup \phi(S')$, and generative effects are sustained.
\end{proof}

Generative effects then occur when there are points in $\phi (S \cup S')$ that do not admit preimages in either $\phi S$ or $\phi S'$.

\section{The problem, in a world that is linear.}

We may push the findings further if we equip our setting with more structure.  We will have both the systems and the phenomes be vector spaces.  If $V$ is a vector spaces, we define $\sub(V)$ to be the lattice of the subspaces of $V$.  The join of $A$ and $B$ in $\sub(V)$ is $A + B$, the linear span of $A$ and $B$.  Let $V$ and $W$ be vector spaces, we consider a veil $\phi : \sub(V) \ra \sub(W)$.

As a concrete example, one may consider $\phi: \sub(\mathbb{R}^n) \ra \sub(\mathbb{R}^{n-1})$ obtained by intersecting a subspace of $\mathbb{R}^n$ by a fixed hyperplane $H$.  One may check that $(\sub(\bR^{n-1}),\phi)$ is a veil for $\sub(\mathbb{R}^n)$ that sustains generative effects for every hyperplane.

If $S$ and $S'$ are subspaces of $V$, we then have two injective linear maps $i : S \ra S + S'$ and $i' : S' \ra S+S'$.  We can then form a linear map $i - i' : S\oplus S' \ra S + S'$.

\begin{Pro}
  For every $S$ and $S'$, the map $i - i'$ is a surjective map.
\end{Pro}

\begin{proof}
 Every element of $S + S'$ can be written in the form $a - (-a')$ with $a \in S$ and $-a' \in S'$.
\end{proof}

On another end, as $\phi S \subset \phi(S + S')$, and $\phi S' \subset \phi(S + S')$ we get linear maps $\phi i :  \phi S \ra \phi(S + S')$, and $\phi i' : \phi S' \ra \phi(S + S')$.  The map $\phi i - \phi i'$ need not always be surjective.

\begin{Pro}\label{Pro:surjA}
 The veil $\phi$ sustain generative effects if, and only if, the map $\phi i - \phi i' : \phi S \oplus \phi S' \ra \phi(S + S')$ is not surjective for some $S$ and $S'$.
\end{Pro}

\begin{proof}
 The same reasoning as that in the proof of Proposition \ref{Pro:surjG} applies, with $\cup$ replaced by $+$.
\end{proof}

If $\phi i - \phi i'$ is not surjective, then there are elements in $\phi(S + S')$ that do not admit a preimage in either $\phi S$ or $\phi S'$. Those points cannot be explained by $\phi S$ and $\phi S'$.

In the linear case, we win an extra characterization of what cannot be explained by the phenome. If $I$ denote the image of $\phi i - \phi i'$, then $\phi(S + S') = I \oplus \phi(S + S')/I$.  What can be explained by the phenome lies in $I$. What cannot be explained by the phenome, and is caused by generative effects, lies in $\phi(S + S')/I$.  Our goal is to characterize and recover $\phi(S + S')/I$.

\subsection{Interlude on exact sequences.}

If we live in a linear world, the relationship among the phenomes of the interconnected system and its separate subsystems will be established through the use of exact sequences.

A sequence of $\bR$-vector spaces $V_i$ and linear maps $f_i$
\begin{equation*}
  \begin{CD}
    \cdots  @>>> V_{i-1} @>f_i>> V_i @>f_{i+1}>> V_{i+1} @>>> \cdots\\
  \end{CD}
  \end{equation*}
is said to be exact at $V_i$ if $\im f_i = \ker f_{i+1}$.  The sequence is said to be exact if it is exact at every $V_i$.  In particular, the sequence:
\begin{equation*}
  \begin{CD} 0  @>>> U @>f>> V \end{CD}
\end{equation*}
is exact if, and only if, the map $f$ is injective.
Dually, the sequence:
\begin{equation*}
  \begin{CD} V  @>g>> W @>>> 0 \end{CD}
\end{equation*}
is exact if, and only if, the map $f$ is surjective.  As we shall see, if the phenomes $\phi(S)$, $\phi(S')$ and $\phi(S + S')$ were made to be part of an exact sequences, we can then relate them together.  For example, if the sequence:
\begin{equation*}
  \begin{CD} U  @>f>> V @>g>> W \end{CD}
\end{equation*}
is exact, then $V$ is isomorphic to $\im f \oplus \ker g$.  If the sequence is longer, the characterization may reach elements further apart.

\subsection{Loss of exactness, on the right.}  

Let $V$ and $W$ be vector spaces. We again consider $\system$ to be $\sub (V)$ and set up a veil $(\sub (W),\phi)$. For $S$ and $S'$ subspaces in $V$, let $j: S \ra S+S'$ and $j': S' \ra S+S'$ be the canonical injections. The sequence:
\begin{equation*}
  \begin{CD} S \oplus S' @>{j - j'}>> S + S' @>>> 0
\end{CD}\end{equation*}
is always exact.  Generative effects occur precisely when:
\begin{equation*}
  \begin{CD} \phi S \oplus \phi S' @>{\phi j-\phi j'}>> \phi (S + S') @>>> 0
\end{CD}\end{equation*}
is not exact. There is then a non-zero vector space $U$ corresponding to $\coker(\phi j- \phi j')$ and a surjective map $\phi (S + S') \ra U$ such that the sequence:
\begin{equation*}
  \begin{CD} \phi S \oplus \phi S' @>{\phi j-\phi j'}>> \phi (S + S') @>>> U @>>> 0
\end{CD}\end{equation*}
is exact.  The vector space $U$ corresponds to the unexplained phenomes, and is isomorphic to $\coker(\phi j- \phi j')$.  We, of course, do not know the map $\phi j- \phi j'$ as we do not know the phenome $\phi (S + S')$. The goal is then to recover $U$ from information on the systems $S$ and $S'$.  As an exemplary approach, we may perform such a recovery via the Snake lemma. 

\begin{Pro}[Snake Lemma, e.g., \cite{ATI1969} ch. 2, p. 23, proposition 2.10]  Given a commutative diagram of vector spaces with exact rows,
\begin{equation*}
\begin{CD}
  0  @>>>    U   @>f>>  V  @>g>>  W @>>> 0\\
  @.      @VV{u}V @VV{v}V @VV{w}V\\ 
  0  @>>> U  @>f'>> V' @>g'>> W' @>>> 0
\end{CD}
\end{equation*}

we get an exact sequence:
\begin{equation*}
  \minCDarrowwidth15pt\begin{CD}
   0 @>>> \ker u @>\tilde{f}>> \ker v @>\tilde{g}>> \ker w @>\delta>> \coker u @>\bar{f'}>> \coker v @>\bar{g'}>> \coker w @>>> 0.
  \end{CD}
\end{equation*}
\end{Pro}
\begin{proof}
  The lemma is standard, and its proof may be found in many texts, e.g., \cite{ATI1969} ch. 2, p. 22. 
\end{proof}

The exact sequence derived form the Snake lemma allows us to link the kernel and cokernel of $w$ to those of $u$ and $v$.
\begin{Pro} \label{Pro:split}
  If the sequence of vector spaces:
\begin{equation*}
  \minCDarrowwidth15pt\begin{CD}
   0 @>>> V_0 @>f>> V_1 @>>> V_2 @>>> V_3 @>g>> V_4 @>>> V_5 @>>> 0,
  \end{CD}
\end{equation*}
is exact, then $V_2 = \coker f \oplus \ker g$ and $V_5 = \coker g$.
\end{Pro}

\begin{proof}
  The following sequence is exact:
  \begin{equation*}
    0 \ra \im (V_1 \ra V_2) \ra V_2 \ra \im (V_2 \ra V_3) \ra 0
  \end{equation*}
  We have $\im (V_1 \ra V_2) = V_1/\ker(V_1 \ra V_2)$.  As $\ker(V_1 \ra V_2) = \im f$ by exactness of the six-term sequence, we get that $\im (V_1 \ra V_2) = \coker(f)$. By exactness, we also get $\im (V_2 \ra V_3) = \ker g$.
  Finally, short exact sequence of vectors spaces split. Namely, if $0 \ra U \ra V \ra W \ra 0$ is a sequence of vector spaces, then $V = U \oplus W$.
\end{proof}

The strategy to recover the phenome coming from generative effects would be to lift the sequence $S \oplus S' \ra S + S' \ra 0$ to be part of the diagram described in the Snake lemma. If we can think of our systems as linear maps, and encode the phenome as a kernel of those maps, then we may recover the map $\phi S \oplus \phi S' \ra \phi (S + S')$ as part of the kernel-cokernel exact sequence.  Furthermore, three columns are in play in the Snake lemma. The middle would correspond to the separate systems. The rightmost column would correspond to the interconnected system.  The leftmost will be made to correspond to the common part of the two systems on which they will be interconnected.\\

As a general insight, let $V$ and $V'$ be subspaces of some vector space, and consider the following commutative diagram:
\begin{equation*}
  \begin{CD}
    V\cap V'  @>i>> V\\
    @VVi'V         @VV{j}V\\
    V'  @>j'>>   V + V'
  \end{CD}
  \end{equation*}

\begin{Pro} \label{Pro:square}
  The sequence
  \begin{equation*}
    \begin{CD}
      0 @>>> V \cap V' @>(i,i')>> V \oplus V' @>j-j'>> V + V' @>>> 0
    \end{CD}
  \end{equation*}
  is exact.
\end{Pro}
\begin{proof}
  The map $(i,i')$ is injective and $j-j'$ is surjective. The map $j-j'$ maps every element $(v,v)$ to $0$.  Conversely, if $j(v) - j'(v') = 0$, then $j(v) = j'(v')$.  As $j$ and $j'$ are injective, we get $v = v'$ and thus $v \in V \cap V'$.
\end{proof}

Interconnection of systems can then be thought of as a short exact sequence.  Generative effects is then equivalent to a loss of exactness, but only on the right.

\begin{Pro}
  The veil $\phi$ sustains generative effects, if and only if, the sequence
 \begin{equation*}
  \begin{CD} 0 @>>> \phi(S \cap S') @>(\phi i,\phi i')>> \phi S \oplus \phi S' @>{\phi j-\phi j'}>> \phi (S + S') @>>> 0
 \end{CD}\end{equation*}
  is not exact at $\phi (S + S')$ for some $S$ and $S'$.  The sequence is always exact at both $\phi (S \cap S')$ and $\phi S \oplus \phi S'$.
\end{Pro}

\begin{proof}
 The map $(\phi i, \phi i')$ is injective, thus the sequence is always exact at $\phi(S \cap S')$. The map ${\phi j-\phi j'}$ maps every element $(v,v)$ to $0$. Conversely, if $\phi j(v)-\phi j'(v') = 0$, then $\phi j(v) = \phi j'(v')$ and so $v = v'$ and belongs to $\phi(S) \cap \phi(S')$.  As $\phi$ is a veil, we have $\phi(S \cap S') = \phi(S) \cap \phi(S')$.  The loss of exactness on the right is characterized by Proposition \ref{Pro:surjA}.
\end{proof}

Through the use of the Snake lemma, we may then relate the phenome of the interconnected system to the separate subsystems and their common part. Of course knowing the phenomes of the separate systems and their common part does not entail us to know the phenome of the interconnected systems.  The cokernels in the exact sequence hold the additional information required to deduce generated phenomes.  The cokernels will encode the generativity of the systems.

\section{The solution, in a world that is not linear.}

We may directly apply the above technique only if the systems are \emph{linear}.  However, most of the settings were are interested in do not possess a linear structure. The goal is to lift our problems, say, to vector spaces.  Such a lift may be for instance achieved by encoding the desired information in the dimension of a vector space.

We will develop in this section a characterization for a simple formulation as a model example. The formulation is the common ground for most of the examples given in Section \ref{Examples}.

\subsection{The formulation.}

Let $S$ be a set $\{e_1,\cdots,e_n\}$ of $n$ elements.  Given an undirected graph $G$ over $S$, we are interested in whether or not there is an undirected path from $e_1$ to $e_n$ in $G$. Neither of the following graphs $G$ or $G'$ contains a path from $e_1$ to $e_3$.
\begin{equation*}
  \begin{tikzcd}[column sep = small, row sep = small]
   G:& a  &  b  \arrow[l,dash] &  c  &&  G':& a   &  b  &  c \arrow[l,dash]
  \end{tikzcd}
\end{equation*}
But if $G$ and $G'$ are combined together, by taking the union of their edge set,
\begin{equation*}
  \begin{tikzcd}[column sep = small, row sep = small]
   G \cup G':& a  &  b  \arrow[l,dash] &  c \arrow[l,dash]
  \end{tikzcd}
\end{equation*}
the edges synchronize and a path emerges.

Let $\G$ denote the semilattice of undirected graphs over $S$, where the join of $G, G' \in \G$ is the graph $G \cup G'$ containing the union of the separate edges.  We define $\phi : \G \ra 2^{\{*\}}$ to be the map such that $\phi G = \{*\}$ if $e_1$ and $e_n$ are connected in $G$ and $\phi G = \emptyset$ otherwise.

The map $\phi$ is not yet a veil, as it does not satisfy V.2.  We can make $\phi$ to be a veil by restricting it to only graphs that are a disjoint union of cliques. Those graphs form a semilattice where the join consists of combining the edges first, then adding edges in each existing connected component to form cliques.  We may then define a system to be a disjoint union of cliques, or equivalently, an equivalence relation on $S$. %
The fix will, however, not affect the solution at all.  The lift we perform will treat both the graph and its \emph{closure} as the same. We may then just ignore the fix, and consider all undirected graphs over $S$ as systems.

\subsection{The lift.}

To lift our problem into a linear world, we let $\bR^S$ denote the free vector space with basis $\{e_1,\cdots,e_n\}$.  If $G \in \G$, we define a subspace $I_G$ of $\bR^S$ to be the span of the vectors $e_i - e_j$ where $\{e_i,e_j\}$ is an edge in $G$.  We then lift every $G$ to a map:
\begin{equation*}
  g: \bR^2 \ra \bR^S/I_G
\end{equation*}
obtained by the composition of the inclusion $\bR^2 \ra \bR^S$ that sends the generators of $\bR^2$ to $e_1$ and $e_n$ in $\bR^S$, and the canonical surjection $\bR^S \ra \bR^S/I_G$.

\begin{Pro}
  The dimension of $\bR^S/I_G$ is equal to the number of connected components in $G$. \qed
\end{Pro}

The kernel of the map, denoted by $\Phi(I_G)$, will then encode the phenome.  For every $G$, we know that $e_1$ and $e_n$ are not in $I_G$. Furthermore, $e_1 - e_n \in I_G$ if, and only if, $e_1$ and $e_n$ are connected in $G$ via a path. We then have: 

\begin{Pro}
  The kernel $\Phi(I_G)$ is isomorphic to $\bR$ if $e_1$ and $e_n$ are connected in $G$ and is the $0$ vector space otherwise. \qed
\end{Pro}

The lift also preserves interconnection of systems.

\begin{Pro}
 If $G, G'\in \G$, then $I_{G \cup G'} = I_G + I_{G'}$. \qed
\end{Pro}

In general, the space $I_G \cap I_G'$ is non-necessarily isomorphic to $I_{G \cap G'}$, even if $G$ and $G'$ are disjoint unions of cliques.  Finally, an inclusion of graphs induces maps on the lifts:
\begin{Pro}
  Let $H$ be a subgraph of $G$, the inclusion graph homomorphism $H \ra G$ lifts to a commutative diagram:
  \begin{equation*}
  \begin{CD}
    \bR^2      @>id>>       \bR^2\\
       @VVhV                            @VVgV\\
    \bR^S/I_H     @>i>>       \bR^S/I_{G}
  \end{CD}
  \end{equation*}
where $i$ is the canonical linear map.
\end{Pro}

\begin{proof}
Let $h : H \ra G$ be an inclusion graph homomorphism. If $\{i,j\}$ is an edge in $H$, then $\{h(i),h(j)\}$ is an edge in $G$.  Thus if $i-j \in I_H$, then $hi - hj \in I_G$.  Therefore $I_H \simeq h I_H \subseteq I_G$. The canonical surjection $\bR^S \ra \bR^S/I_{G}$ then factors through $\bR^S/I_{H}$ to yield $i$. 
\end{proof}

The connected components of a graph $G$ form a basis for $\bR^S/I_{G}$. As $H$ is a subgraph of $G$, the map $i: \bR^S/I_H \ra \bR^S/I_{G}$ is surjective and sends connected components of $H$ to connected components of $G$ in a manner compatible with the inclusion.

\subsection{Recovering exactness.}

Those square diagrams can then be neatly fitted in an exact diagram:

\begin{Pro} The following diagram is commutative and has exact rows:
\begin{equation*}
\begin{CD}
  0  @>>>      \bR^2                 @>>>  \bR^2 \oplus \bR^2        @>>>      \bR^2                   @>>> 0\\
  @.             @VVV                                 @V(g,g')VV                            @VVV                            \\ 
  0  @>>>    \bR^S/(I_G \cap I_{G'})   @>(i,i')>> \bR^S/I_G \oplus \bR^S/I_{G'}    @>j - j'>>      \bR^S/(I_G + I_{G'})      @>>> 0
\end{CD}
\end{equation*}
\end{Pro}

\begin{proof}
  To show exactness of the bottom row, apply the Snake lemma to the canonical diagram:
  \begin{equation*}
    \begin{CD}
  0  @>>>    I_G \cap I_{G'}   @>>> I_G \oplus I_{G'}    @>>>      I_G + I_{G'}      @>>> 0\\
 @.             @VVV                                 @VVV                            @VVV                            \\ 
  0  @>>>    \bR^S  @>>> \bR^S \oplus \bR^S    @>>>      \bR^S     @>>> 0
  \end{CD}
    \end{equation*}  
  whose upper row we know is exact from Proposition \ref{Pro:square}.
\end{proof}

We may then recover an exact sequence from the Snake lemma. We first summarize the pieces of the sequence. 
  Every square diagram:
    \begin{equation*}
  \begin{CD}
    \bR^2      @>id>>       \bR^2\\
       @VVhV                            @VVgV\\
    \bR^S/I_H     @>i>>       \bR^S/I_{G}
  \end{CD}
    \end{equation*}
    can be extended to a commutative diagram with exact columns:
   \begin{equation*}
    \begin{CD}
      0           @.                 0\\
       @VVV                        @VVV\\
      \Phi(I_H)   @>\Phi(i)>>      \Phi(I_G)\\
      @VVV                        @VVV\\
    \bR^2      @>id>>       \bR^2\\
       @VVhV                          @VVgV\\
       \bR^S/I_H     @>i>>       \bR^S/I_{G}\\
       @VVV                   @VVV\\
      \bH(I_H) = \bR^S/(I_H + \<e_1,e_n\>)     @>\bH(i)>>       \bH(I_G) = \bR^S/(I_{G}  + \<e_1,e_n\>)\\
       @VVV                        @VVV\\
       0          @.                 0
  \end{CD}
 \end{equation*}
The space $\<e_1,e_n\>$ is the subspace of $\bR^S$ generated by $e_1$ and $e_n$.  The vector space $\bR^S/(I_G + \<e_1,e_n\>)$, denoted by $\bH(I_G)$, is the cokernel of $g$, and the map $\bH(i):\bH(I_H) \ra \bH(I_G)$ sends an element $a + I_H + \<e_1,e_n\>$ in $\bH(I_H)$ to $i(a + W_H) + \<e_1,e_n\>$ in $\bH(I_G)$. \\

The dimension of the space $\bH(I_G)$ is equal to the number of connected components in $G$ not containing either $e_1$ or $e_n$.  The map $\bH(i)$ then destroys all the components in $H$ that do not contain $e_1$ or $e_n$ in $H$ but that do contain them in the image component.\\

Finally, if $G, G' \in \G$, we recover an exact sequence.
\begin{equation*}
  \begin{tikzcd}[column sep = large]
    0  \arrow[r] & \Phi(I_G \cap I_{G'}) \arrow[r,"{(\Phi(i),\Phi(i'))}"] &  \Phi I_G \oplus \Phi I_{G'} \arrow[r,"{\Phi (j)-\Phi (j')}"] &  \Phi (I_{G \cup G'}) \arrow[dll] & \\
                 &  \bH(I_G \cap I_{G'}) \arrow{r}[swap]{(\bH (i),\bH (i'))} & \bH I_G \oplus \bH I_{G'} \arrow{r}[swap]{\bH (j)-\bH (j')} &  \bH (I_{G \cup G'}) \arrow[r] & 0
  \end{tikzcd}
\end{equation*}
\begin{Pro}
  We have: \[\Phi (I_{G \cup G'}) = \coker((\Phi(i),\Phi(i'))) \oplus \ker ((\bH (i),\bH (i')))\] and: \[\bH (I_{G \cup G'}) = \coker ((\bH (i),\bH (i'))).\]
\end{Pro}
\begin{proof}
  Apply Proposition \ref{Pro:split} to the six-term exact sequence.
\end{proof}

The space $\coker((\Phi(i),\Phi(i')))$ encodes whether or not $e_1$ and $e_n$ are connected in one of the separate graphs, and the space $\ker((\bH (i),\bH (i')))$, or equivalently $\ker(\bH({i})) \cap \ker(\bH({i'}))$, encodes the formation of such a path via generative effects. In particular,
\begin{Pro}
 If $\Phi(I_G) = \Phi(I_{G'}) = 0$, then $\Phi(I_{G\cup G'}) = \ker(\bH({i})) \cap \ker(\bH({i'}))$. \qed
\end{Pro}

As an explicit characterization, we have:
\begin{equation*}
  \ker((\bH (i),\bH (i'))) = \faktor{(I_G + \<e_1,e_n\>)\cap (I_{G'} + \<e_1,e_n\>)}{I_G \cap I_{G'} + \<e_1,e_n\>}
\end{equation*}
Whether or not a path is created is then encoded in the difference of the dimensions of $(I_G + \<e_1,e_n\>)\cap (I_{G'} + \<e_1,e_n\>)$ and $I_G \cap I_{G'} + \<e_1,e_n\>$.  Such a discrepancy would exist as the lattice of subspaces in a not a distributive lattice, but only a modular one.  This remark will not be further pursued.  The explicit characterization provided will however be further considered in a later section dealing with graphs in multiple universa. 

The vector spaces $\bH(I_.)$ also admit additional structure that allows them to be interpreted combinatorially.

\subsection{A combinatorial criterion.}

The graphs, in this section, are assumed, for simplicity, to be disjoint union of cliques.  
Let $\M(G) \subseteq 2^S$ be the collection of non-empty subsets $V \subseteq S$ of vertices, where the induced subgraph of $G$ on $V$ has a perfect matching.
The vector space $I_G$ can be derived from the set $\M(G)$.  %
The vector space $I_G \cap I_{G'}$ can be derived from $\M(G) \cap \M(G')$.

\begin{Def}
  If $\M \subseteq 2^S$, we define:
\begin{equation*}
  \bH \M = \bigg\{ \{e_1\}, \{e_n\}\bigg\} \cup \bigg\{V - \{e_1, e_n\} : V \in \M \text{ is minimal}\bigg\}. 
\end{equation*}
\end{Def}
The operation $\bH$ \emph{splits} the minimal elements of $\M$ whenever they contain either $e_1$ or $e_n$.  The set $\bH \M(G)$ then corresponds to the induced subgraph of $G$ on the nodes not connected to $e_1$ or $e_n$.  The splitting procedure, thus, destroys information in $\M$. For instance, if $G$ and $G'$ are the graphs:
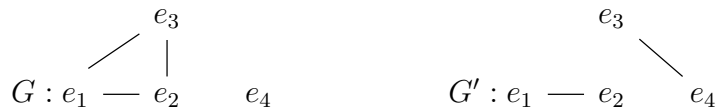
\begin{equation*}
  \begin{tikzcd}[column sep = small, row sep = small]
      & e_3 \arrow[dl,dash] \arrow[d,dash] &   &&&&      & e_3 \arrow[dr,dash] &         \\
        G: e_1 \arrow[r,dash] & e_2 & e_4   &&&&    G': e_1 \arrow[r,dash] & e_2 & e_4 
  \end{tikzcd}
\end{equation*}
then $\bH \M(G) = \bH \M(G') = \big\{\{e_1\},\{e_2\},\{e_3\},\{e_4\}\big\}$, although $\M(G)$ and $\M(G')$ are different. 

\begin{Pro}
  If $\Phi(I_G) = \Phi(I_{G'}) = 0$, then $\Phi(I_{G \cup G'}) = \bR$ if, and only if, there exists a set $V \subseteq S$ such that:
  \begin{itemize}
  \item[C.1.] The set $V$ is a disjoint union of sets in $\bH \M(G)$,
  \item[C.1'.] The set $V$ is a disjoint union of sets in $\bH \M(G')$,
  \item[C.2.]  The set $V$ is not a disjoint union of sets in $\bH (\M(G) \cap \M(G'))$.  
  \end{itemize}
\end{Pro}

\begin{proof} 
  Let $V$ be a minimal such set.  As C.1 (resp. C.1') holds, the nodes in $V$ that are not connected to either $e_1$ or $e_n$ in $G$ (resp. $G'$) can be perfectly matched in $G$ (resp. $G'$).  As C.2 holds, by minimality of $V$, no two nodes in $V$ share the same component in both $G'$ and $G$.  Furthermore, as C.2 holds, no subset in $V$ can belong to a cycle in $G \cup G'$.  The nodes in $V$ then have to form a path between $e_1$ and $e_n$ in $G \cup G'$.
  
  Conversely, if $\Phi(I_{G \cup G'}) = \bR$, then $\ker \bH(i)$ and $\ker \bH(i')$ have a common element, an alternating sum $\sum_k (-1)^ke_{i_k}$ for $e_{i_k} \in S$ a vertex and $0 \leq k \leq m-1$.  Pick the common element consisting of the least number of vertices.  This set of vertices forms a set $V$ satisfying C.1, C.1' and C.2.
\end{proof}

\subsection{Concrete instances.}
The answer to the following three questions is obviously yes.  We answer them to illustrate the workings of the above theory.

\paragraph{Q.1.} If we combine $G$ and $H$, do we get a path from $a$ to $c$?
\begin{equation*}
  \begin{tikzcd}[column sep = small, row sep = small]
   G:& a  &  b  \arrow[l,dash] &  c  &&  H:& a   &  b  &  c \arrow[l,dash]
  \end{tikzcd}
\end{equation*}
We have $I_G = \<a-b\>$ and $I_H = \<b-c\>$, and thus $I_G \cap I_H = 0$.  We then get $\bH(i) : \bR \ra 0$ and $\bH(i') : \bR \ra 0$.  Clearly then $\Phi(I_{G\cup G}) = \ker(\bH(i)) \cap \ker(\bH(i')) = \bR$.

Combinatorially, we have:
\[\bH \M(G) = \bH \M(H) = \big\{\{a\},\{b\},\{c\}\big\}.\]
The set $\bH \M(G) \cap \M(H)$ is $\big\{\{a\}.\{c\}\big\}$ as no subset supports a perfect matching in both graphs.  The set $\{b\}$ satisfies the required conditions.

\paragraph{Q.2.} If we combine $G$ and $H$, do we get a path from $a$ to $e$?
\begin{equation*}
  \begin{tikzcd}[column sep = small, row sep = small]
    G:& a  &  b  \arrow[l,dash] & c & d \arrow[l,dash] & e    &&  H:& a   &  b  &  c \arrow[l,dash] & d & e \arrow[l,dash]
  \end{tikzcd}
\end{equation*}
We have $I_G = \< a-b, d-c\>$ and $I_H = \<b-c, d-e\>$, and thus $I_G \cap I_H = 0$. We get $\bH(i) : \bR^3 \ra \bR$ and $\bH(i') : \bR^3 \ra \bR$ where $\bR^3$ is generated by $b$, $c$ and $d$.  The map $\bH(i)$ sends both $b$ and $c-d$ to $0$.  The map $\bH(i')$ sends both $d$ and $b-c$ to $0$.  The element $b - c + d$ then generates the intersection of the kernels.  We then get $\Phi(I_{G \cup H}) = \bR$.

Combinatorially, we have:
\[\bH \M(G) = \big\{\{a\},\{b\},\{c,d\},\{e\}\big\} \quad and \quad \bH \M(H) = \big\{\{a\},\{b,c\},\{d\},\{e\}\big\}.\]
The set $\bH \M(G) \cap \M(H)$ is $\big\{\{a\}.\{e\}\big\}$ as no subset supports a perfect matching in both graphs.  The set $\{b,c,d\}$ satisfies the required conditions.

\paragraph{Q.3.} If we combine $G$ and $H$, do we get a path from $a$ to $d$?
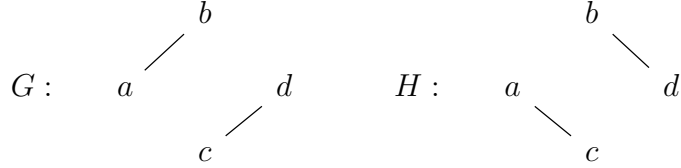
\begin{equation*}
  \begin{tikzcd}[column sep = small, row sep = small]
    & & b \arrow[dl,dash] &    &&      & & b \arrow[dr,dash] &         \\
   G: & a  &  & d  \arrow[dl,dash]    &&    H:&  a  \arrow[dr,dash] &  & d      \\
   & & c &   &&  &&c& 
  \end{tikzcd}
\end{equation*}
We have $I_G = \<a-b,c-d\>$ and $I_H = \<a-c, b-d\>$, and then $I_G \cap I_H = \<a - b + d - c\>$.  We get $\bH(i) : \bR \ra 0$ and $\bH(i') : \bR \ra 0$. Clearly then $\Phi(I_{G\cup G}) = \ker(\bH(i)) \cap \ker(\bH(i')) = \bR$.

Combinatorially, we have:
\[\bH \M(G) = \bH \M(H) = \big\{\{a\},\{b\},\{c\},\{d\}\big\}.\]
The set $\bH \M(G) \cap \M(H)$ is $\big\{\{a\},\{b,c\},\{d\}\big\}$. Both the sets $\{b\}$ and $\{c\}$ satisfy the required conditions.

\subsection{Encoding generativity.}

We have thus related the phenome of the combined graph to the separate graphs through the objects $\bH(I)$ via an exact sequence:
\begin{equation*}
  \begin{tikzcd}[column sep = large]
    0  \arrow[r] & \Phi(I_G \cap I_{G'}) \arrow[r,"{(\Phi(i),\Phi(i'))}"] &  \Phi (I_G) \oplus \Phi (I_{G'}) \arrow[r,"{\Phi (j)-\Phi (j')}"] &  \Phi (I_{G \cup G'}) \arrow[dll] & \\
                 &  \bH(I_G \cap I_{G'}) \arrow{r}[swap]{(\bH (i),\bH (i'))} & \bH (I_G) \oplus \bH (I_{G'}) \arrow{r}[swap]{\bH (j)-\bH (j')} &  \bH (I_{G \cup G'}) \arrow[r] & 0
  \end{tikzcd}
\end{equation*}
The objects $\bH(I_.)$ may be then be seen to encode at least what is essential explain generative effects. As such, if the objects $\bH(I_.)$ are always $0$, one then concludes that generative effects are not sustained. As a rough converse, the objects $\bH(I_G)$ may be seen to encode only what is essential for generative effects in the system.  Indeed, we did not use all the information of the separate graphs. For instance, the problem of combining any of the two graphs:
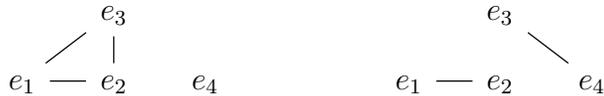
\begin{equation*}
  \begin{tikzcd}[column sep = small, row sep = small]
      & e_3 \arrow[dl,dash] \arrow[d,dash] &   &&&&      & e_3 \arrow[dr,dash] &         \\
        e_1 \arrow[r,dash] & e_2 & e_4   &&&&    e_1 \arrow[r,dash] & e_2 & e_4     
  \end{tikzcd}
\end{equation*}
with the graph:
\begin{equation*}
  \begin{tikzcd}[column sep = small, row sep = small]
               & e_3 &      \\
    e_1        & e_2 & e_4 \arrow[l,dash] 
  \end{tikzcd}
\end{equation*}
will yield the same exact sequence from the Snake Lemma.  The special theory developed in this paper is not, however, set up to discuss well how much of the information is kept from the system.  One however ought to expect a good variation.  In some cases, a large amount of information is irrelevant and will be brushed away by the $\bH(I)$ objects. On another end, one may devise example where almost everything from the system is fundamentally bound to play a part in generating effects.  Such an example may go along the lines of Example \ref{warmupExample}.  In such situations, the generativity of the system tends to get close to exactly what is concealed under the veil.

\subsection{For graphs living on different vertices.}

The graphs $G$ and $G'$ have both been defined over the same vertex set $S$.  The same problem can be more generally recast by gluing two graphs defined over different vertices over a common subsets of vertices. The problem is formally set up as follows.

Let $G$ and $G'$ be undirected graphs over $V$ and $V'$ respectively, and let $C$ be a set. The set $C$ is to be interpreted as the set of common vertices.  As such we are given inclusions:
\begin{equation*}
  i : C \ra V \quad \text{ and } \quad i' : C \ra V'.
\end{equation*}
The nodes $i(v)\in V$ and $i'(v)\in V'$ will be identified as the same vertices, to form the glued graph.  The gluing construction mathematically amounts to taking a pushout from $C$ along $i$ and $i'$.  Pushouts however are outside the scope of this paper, and we thus revert to an algorithmic construction.  We form an undirected graph $G^*$ over a set $V^*$.  If $n$, $n'$ and $c$ denote the cardinality of $V$, $V'$ and $C$, then $V^*$ has cardinality $n + n' - c$. The set $V^*$ will then be seen to contain both $V$ and $V'$. As such we have two inclusions:
\begin{equation*}
  j: V \ra V^* \quad \text{ and } \quad j': V' \ra V^*,
\end{equation*}
whose images coincide on $C$.  The edges of $G^*$ are $j(u)\sim j(v)$ whenever $u \sim v$ in $G$ and $j'(u')\sim j'(v')$ whenever $u' \sim v'$ in $G'$.

We then pick two distinguished vertices $s$ and $t$.  Each of the vertices $s$ and $t$ may either lie in $V$, in $V'$ or in both (i.e., in $C$).  The question then is:

\begin{Que}
  Given $G$ and $G'$ to be glued along a vertex set $C$, and a compatible choice of $s$ and $t$, is there an undirected path from $s$ to $t$ in $G^*$?
\end{Que}
If the set $C$ is equal to $V$ and $V'$ (i.e., $i$ and $i'$ are bijections), we then recover the original formulation of the situation, whereby the graphs are defined over the same vertex set.

\subsubsection{The characterization.}

We get different cases, depending on whether or not $s$ and $t$ lie in $V$ or $V'$.  All the cases can be dealt with mathematically in an implicit manner.  We will however deal with some explicitely for clarity of exposition, and present the rest in a generic form.  We however omit proofs and refer the reader to Chapter 8 of \cite{ADAM:Dissertation} for more details.  

Let $n$, $n'$ and $c$ be the cardinality of $V$, $V'$ and $C$.  The inclusion $i : C \ra V$ (resp. $i : C \ra V'$) induces an injective linear map $\iota : \bR^{c} \ra \bR^{n}$ (resp. $\iota' : \bR^{c} \ra \bR^{n'}$). If $I$ is a subspace of $\bR^n$ (resp. of $\bR^{n'}$), we define $\pi(I)$ (resp. $\pi'(I)$) to be $\{ a \in R^{c} : \iota a \in I\}$ (resp. $\pi'(I) = \{ a \in R^{c} : \iota' a \in I\}$.  Note that $\pi(I)$ and $\pi'(I)$ are both subspaces of $\bR^{c}$.

We will suppose that no path from $s$ to $t$ already exists in $G$ or in $G'$, separately.  That could either be due to the fact that either $s$ or $t$ is not in $V$ (or in $V'$), or that the edges, in the separate graphs, simply do not \emph{synchronize} to produce a path.

\paragraph{Case 1.}

We consider the case where $s$ and $t$ are in $G$ but not in $G'$.  A path is then created if, and only if:
\begin{equation*}
  \faktor{\pi(I_G + \<s,t\>) \cap \pi'(I_{G'})}{  \pi(I_G) \cap \pi'(I_{G'})} \neq 0 
\end{equation*}
Checking the presence of such an inequality amounts to comparing the dimensions of ${\pi(I_G + \<s,t\>) \cap \pi'(I_{G'})}$ and ${\pi(I_G) \cap \pi'(I_{G'})}$.  Note that these are subspaces of $\bR^{|C|}$, i.e. the computation is performed only over the common nodes. Furthermore, by fixing $G$ and $C$, we can vary $G'$ only computing $\pi'(I_{G'})$, without having to recompute additional information on $G$.

\paragraph{Case 2.}

We consider the case where $s$ is in $G$ but not in $G'$ and $t$ is in both. A path is then created if, and only if:
\begin{equation*}
  \faktor{\pi(I_G + \<s,t\>) \cap \pi'(I_{G'} + \<t\>)}{  \pi(I_G) \cap \pi'(I_{G'}) + \<t\>} \neq 0 
\end{equation*}

\paragraph{Case 3.}

We consider the case where $s$ is in $G$ but not in $G'$ and $t$ is in $G'$ but not in $G$.  A path is then created if, and only if:
\begin{equation*}
  \faktor{\pi(I_G + \<s\>) \cap \pi'(I_{G'} + \<t\>)}{  \pi(I_G) \cap \pi'(I_{G'})} \neq 0 
\end{equation*}

\paragraph{All other cases.}

In the general case, a path is created if, and only if:
\begin{equation*}
  \faktor{\pi(I_G + J_G) \cap \pi'(I_{G'} + J_{G'})}{  \pi(I_G) \cap \pi'(I_{G'}) + J_{G^*}} \neq 0 
\end{equation*}
where:

\begin{equation*}
  J_G = \left\{
  \begin{array}{ll}
    \<s\> & \text{if } s\in V \text{ and } t\notin V\\
    \<t\> & \text{if } s\notin V \text{ and } t\in V\\
    \<s,t\> & \text{if } s\in V \text{ and } t\in V\\
    0 & \text{if } s\notin V \text{ and } t\notin V
  \end{array}\right.
\end{equation*}

\begin{equation*}
  J_{G'} = \left\{
  \begin{array}{ll}
    \<s\> & \text{if } s\in V' \text{ and } t\notin V'\\
    \<t\> & \text{if } s\notin V' \text{ and } t\in V'\\
    \<s,t\> & \text{if } s\in V' \text{ and } t\in V'\\
    0 & \text{if } s\notin V' \text{ and } t\notin V'
  \end{array}\right.
\end{equation*}

\begin{equation*}
  J_{G*} = \left\{
  \begin{array}{ll}
    \<s\> & \text{if } s\in C \text{ and } t\notin C\\
    \<t\> & \text{if } s\notin C \text{ and } t\in C\\
    \<s,t\> & \text{if } s\in C \text{ and } t\in C\\
    0 & \text{if } s\notin C \text{ and } t\notin C
  \end{array}\right.
\end{equation*}

 We refer the reader to Chapter 8 of \cite{ADAM:Dissertation} for the details and some proofs on such characterizations. By defining bases, the dimensions may be computed using matrix operations.  Those operations could also lend themselves to combinatorial interpretations given the nature of the problem.  The direction will not be further pursued in this paper.

\section{Concluding remarks.}

This paper did not explicitely expound most of the mathematical connections that arise throughout.  Its goal was to present the theory with the least amount of diversion possible. We end with three general remarks.

The maps between the systems, and their lifts, played an important role in the characterization.  They did not, however, explicitely appear when interaction and generativity were initially defined.  They were, nevertheless, always implicit in the partial order on the semilattice.  The general level of the theory is then achieved by explicitely defining maps, or morphisms, between systems. In the general theory, those morphisms will be as important as (if not more important than) the systems themselves.

The lifts in this paper have been described through vector spaces.  Different linear (or \emph{abelian}) objects ought to, however, be used to capture richer structures in the semilattice of systems and veils.  The notion of exact sequences, and the results of the snake lemma, will remain in effects.  Recovering information from the exact sequence will, however, not be as straightforward as it is in the case of vector spaces. It might require extra information from the systems.  Our use of the snake lemma is, furthermore, only a special example of a mechanism relating the phenome of the combined system to the separate subsystems.

Finally, the tractabilty of the problem---coping with generative effects sustained by the veil---will depend on the tractability of the lifts. The same problem may possess different lifts.  The \emph{better} the lift is in capturing the structure of the problem, the \emph{better} the linking solution is.

\bibliographystyle{alpha}
\bibliography{MyBiblio}

\section{Appendix: On computability.}

The ideas presented in this paper, and the algebraic machinery, can be used to derive computable criteria for phenomes to emerge.  As a certificate, we implemented a verbose program (in \texttt{GAP v.4.8.3}) of the path example. Below is a snapshot of the output of the program.  The output is somewhat self-explanatory.  The program takes two graphs $G$ and $H$ as input.  The common nodes are then specified in both $G$ and $H$.  The endpoints $A$ and $E$ are also specified accordingly.  The program then outputs whether or not there exists a path from $A$ to $E$ once the common nodes are identified.  The answer is obtained via the characterization provided in the earlier section, on graphs living on different vertices.

\begin{verbatim}
-------------------------
 Info of G:
-------------------------
 Nb of nodes: 3
 Nb of edges: 1
 List of Edges: [ [ 1, 2 ] ]
 Nb of Components: 2
 Dim of H^1(G): 0

 No path in G.

-------------------------
 Info on H:
-------------------------
 Nb of nodes: 3
 Nb of edges: 1
 List of Edges: [ [ 2, 3 ] ]
 Nb of Components: 2
 Dim of H^1(H): 0

 No path in H.

-------------------------
 Info on Common:
-------------------------
 Nb of nodes: 3
 IDs in G: [ 1, 2, 3 ]
 IDs in H: [ 1, 2, 3 ]

-------------------------
 Distinguished nodes:
-------------------------
 A value of -1 indicates non-existence.
 A in G @: 1
 E in G @: 3
 A in H @: 1
 E in H @: 3

-------------------------
 Pulling Back
-------------------------
 Dim of pullback from G: 2
 Dim of pullback from H: 2
 Dim of Common System: 3
 Dim of H^1(Common): 1

-------------------------
 Decision Criterion:
-------------------------
 Dim of piH^1(G): 0
 Dim of piH^1(H): 0
 Dim of Kernel of F^3 --> piH^1(G)(+)piH^1(H): 3
 Quotienting subspace in H^1(Common): 2

 The program computes 4 objects:

 Dim of piI_G: 1
 Dim of piI_H: 1
 Dim of intersection: 0
 Augmented Dimension: 2    (1)

 Dim of pi(I_G+<A,E>): 3
 Dim of pi(I_H+<A,E>): 3
 Dim of intersection: 3    (2)

 Compare (1) and (2):  path emerges iff different.

#########################
 Path exists!
#########################
 
 Path emerged!
 
\end{verbatim}

Modifying the graph $H$ yields:

\begin{verbatim}
-------------------------
 Info of G:
-------------------------
 Nb of nodes: 3
 Nb of edges: 1
 List of Edges: [ [ 1, 2 ] ]
 Nb of Components: 2
 Dim of H^1(G): 0

 No path in G.

-------------------------
 Info on H:
-------------------------
 Nb of nodes: 3
 Nb of edges: 1
 List of Edges: [ [ 1, 2 ] ]
 Nb of Components: 2
 Dim of H^1(H): 0

 No path in H.

-------------------------
 Info on Common:
-------------------------
 Nb of nodes: 3
 IDs in G: [ 1, 2, 3 ]
 IDs in H: [ 1, 2, 3 ]

-------------------------
 Distinguished nodes:
-------------------------
 A value of -1 indicates non-existence.
 A in G @: 1
 E in G @: 3
 A in H @: 1
 E in H @: 3

-------------------------
 Pulling Back
-------------------------
 Dim of pullback from G: 2
 Dim of pullback from H: 2
 Dim of Common System: 2
 Dim of H^1(Common): 0

-------------------------
 Decision Criterion:
-------------------------
 Dim of piH^1(G): 0
 Dim of piH^1(H): 0
 Dim of Kernel of F^3 --> piH^1(G)(+)piH^1(H): 3
 Quotienting subspace in H^1(Common): 3

 The program computes 4 objects:

 Dim of piI_G: 1
 Dim of piI_H: 1
 Dim of intersection: 1
 Augmented Dimension: 3    (1)

 Dim of pi(I_G+<A,E>): 3
 Dim of pi(I_H+<A,E>): 3
 Dim of intersection: 3    (2)

 Compare (1) and (2):  path emerges iff different.

#########################
 No Path.
#########################

\end{verbatim}

Finally, as an example of graphs defined on different sets of vertices, we get:

\begin{verbatim}
-------------------------
 Info of G:
-------------------------
 Nb of nodes: 13
 Nb of edges: 10
 List of Edges: [ [ 3, 5 ], [ 1, 7 ], [ 6, 7 ], [ 6, 12 ], 
       [ 12, 2 ], [ 10, 8 ], [ 8, 9 ], [ 11, 9 ], [ 9, 4 ], [ 13, 4 ] ]
 Nb of Components: 3
 Dim of H^1(G): 2
 
 No path in G.

-------------------------
 Info on H:
-------------------------
 Nb of nodes: 17
 Nb of edges: 9
 List of Edges: [ [ 2, 4 ], [ 1, 5 ], [ 5, 13 ], [ 3, 7 ],
       [ 7, 9 ], [ 9, 15 ], [ 11, 9 ], [ 11, 13 ], [ 6, 10 ] ]
 Nb of Components: 8
 Dim of H^1(H): 7

 No path in H.

-------------------------
 Info on Common:
-------------------------
 Nb of nodes: 5
 IDs in G: [ 1, 2, 3, 4, 5 ]
 IDs in H: [ 1, 2, 3, 4, 5 ]

-------------------------
 Distinguished nodes:
-------------------------
 A value of -1 indicates non-existence.
 A in G @: 10
 E in G @: -1
 A in H @: -1
 E in H @: 9

-------------------------
 Pulling Back
-------------------------
 Dim of pullback from G: 3
 Dim of pullback from H: 2
 Dim of Common System: 4
 Dim of H^1(Common): 4

-------------------------
 Decision Criterion:
-------------------------
 Dim of piH^1(G): 2
 Dim of piH^1(H): 1
 Dim of Kernel of F^5 --> piH^1(G)(+)piH^1(H): 2
 Quotienting subspace in H^1(Common): 1

 The program computes 4 objects:

 Dim of piI_G: 2
 Dim of piI_H: 3
 Dim of intersection: 1
 Augmented Dimension: 1    (1)

 Dim of pi(I_G+<A,E>): 3
 Dim of pi(I_H+<A,E>): 4
 Dim of intersection: 2    (2)

 Compare (1) and (2):  path emerges iff different.

#########################
 Path exists!
#########################
 
 Path emerged!

\end{verbatim}

\end{document}